\documentclass{amsart}

\usepackage{amsfonts,amstext,amsmath,amssymb,stmaryrd,mathrsfs,calc,amsthm,xspace,mathtools}
\usepackage{bm}
\usepackage{comment}
\usepackage{framed}
\usepackage{graphics}
\usepackage{enumerate}
\usepackage{todonotes}
\usepackage{mathtools}
\usepackage{complexity}

\usepackage{xspace,tabularx,multirow}
\usepackage{color}
\usepackage{colortbl}

\usepackage{macros}
\usepackage{hyperref}

\usepackage{lipsum}

\usepackage{tikz}
\usepackage[normalem]{ulem}
\usepackage{ifthen}
\usetikzlibrary{arrows, automata,calc,shapes.arrows, backgrounds, fadings,intersections,  decorations.markings, positioning, snakes}
\usetikzlibrary{shadows}
\tikzfading[name=fade out, inner color=transparent!0,
  outer color=transparent!100]
\tikzset{
    every picture/.style={>=stealth,auto,node distance=2cm,}
}
\tikzstyle{every state}=[
    inner sep=0pt,
    minimum size=0.7cm,
    fill= white,
    fill opacity=0.2,
    draw= black,
    text= black,
    text opacity=1,
    scale=0.5,
    ]

\title[Equivalence Problems for
  Commutative Grammars]{Tightening The Complexity of Equivalence
  Problems for Commutative Grammars$^*$} \thanks{$^*$Supported by
  Labex Digicosme, Univ.\ Paris-Saclay, project VERICONISS}

\author{Christoph Haase}
\author{Piotr Hofman}
\address{Laboratoire Sp\'ecification et V\'erification (LSV), CNRS
  \& ENS de Cachan, France}

\begin{document}

\begin{abstract}
  We show that the language equivalence problem for regular and
  context-free commutative grammars is \coNEXP-complete. In addition,
  our lower bound immediately yields further \coNEXP-completeness
  results for equivalence problems for communication-free Petri nets
  and reversal-bounded counter automata. Moreover, we improve both
  lower and upper bounds for language equivalence for
  exponent-sensitive commutative grammars.
\end{abstract}

\maketitle

\section{Introduction}

Language equivalence is one of the most fundamental decision problems
in formal language theory. Classical results
include \PSPACE-completeness of deciding language equivalence for
regular languages generated by non-deterministic finite-state automata
(NFA)~\cite[p.~265]{GJ79}, and the undecidability of language
equivalence for languages generated by context-free
grammars~\cite[p.~318]{HMU03}.

Equivalence problems for formal languages which are undecidable over
the free monoid may become decidable in the commutative setting. The
problem then is to decide whether the Parikh images of two languages
coincide. Given a word $w$ over an alphabet $\Sigma$ consisting of $m$
alphabet symbols, the Parikh image of $w$ is a vector in $\N^m$
counting in its $i$-th component how often the $i$-th alphabet symbol
occurs in $w$. This definition can then be lifted to languages, and
the Parikh image of a language consequently becomes a subset of
$\N^m$, or, equivalently, a subset of $\Sigma^\odot$, the free
commutative monoid generated by $\Sigma$. Parikh's theorem states that
Parikh images of context-free languages are semi-linear sets. Since
the latter are closed under all Boolean operations~\cite{Gin66},
deciding equivalence between Parikh images of context-free languages
is decidable.

When dealing with Parikh images of formal languages, it is technically
more convenient to directly work with commutative grammars, which were
introduced by Huynh in his seminal paper~\cite{Huy83} and are
``generating devices for commutative languages [that] use [the] free
commutative monoid instead of [the] free monoid.'' In~\cite{Huy83},
Huynh studied the uniform word problem for various classes of
commutative grammars; the complexity of equivalence problems for
commutative grammars was subsequently investigated in a follow-up
paper~\cite{Huy85}. One of the main results in~\cite{Huy85} is that
the equivalence problem for regular and context-free commutative
grammars is $\ComplexityFont{\Pi_2^{\P}}$-hard and in \coNEXP. Huynh
remarks that a better upper bound might be possible, and states as an
open problem the question whether the equivalence problem for
context-free commutative grammars is
$\ComplexityFont{\Pi_2^\P}$-complete~\cite[p.~117]{Huy85}. Some
progress towards answering this question was made by Kopczy{\'n}ski
and To, who showed that inclusion and \emph{a fortiori} equivalence
for regular and context-free commutative grammars is
$\ComplexityFont{\Pi_2^\P}$-complete when the size of the alphabet is
fixed~\cite{KT10,Kop15}. One of the main contributions of this paper
is to answer Huynh's question negatively: we show that already for
regular commutative grammars the equivalence problem is
\coNEXP-complete. 

Our \coNEXP\ lower bound is established by showing how to reduce
validity in the \coNEXP-complete $\Pi_2$-fragment of Presburger
arithmetic~\cite{Gra89,Haa14} (i.e. its $\forall^*\exists^*$-fragment)
to language inclusion for regular commutative grammars. A reduction
from this fragment of Presburger arithmetic has recently been used
in~\cite{HH14} in order to show \coNEXP-completeness of inclusion for
integer vector addition systems with states ($\mathbb{Z}$-VASS), and
this reduction is our starting point. Similarly to the standard
definition of vector addition systems with states, $\mathbb{Z}$-VASS
comprise a finite-state controller with a finite number of counters
which, however, range over the integers. Consequently, counters can be
incremented and decremented, may drop below zero, and the order in
which transitions in $\mathbb{Z}$-VASS are taken may commute along a
run---those properties are crucial to the hardness proof
in~\cite{HH14}. The corresponding situation is different and
technically challenging for regular commutative grammars. In
particular, alphabet symbols can only be produced but not deleted,
and, informally speaking, we cannot produce negative quantities of
alphabet symbols.

A further contribution of our paper is to establish a new upper bound
for the equivalence problem for exponent-sensitive commutative
grammars, a generalisation of context-free commutative grammars where
the left-hand sides of productions may contain an arbitrary number of
some non-terminal symbol. Exponent-sensitive commutative grammars were
recently introduced by Mayr and Weihmann in~\cite{MW13}, who
showed \PSPACE-completeness of the word problem, and membership in
\ComplexityFont{2}-\EXPSPACE\ of the equivalence problem. Our hardness
result implies that the equivalence problem is \coNEXP-hard, and we
also improve the \ComplexityFont{2}-\EXPSPACE-upper bound to
\co-$\ComplexityFont{2}\NEXP$.

Finally, commutative grammars are very closely related to Petri nets,
c.f.~\cite{Huy83,Esp97,Yen97,MW15}. We also discuss implications of
our results to equivalence problems for various classes of Petri nets
as well as reversal-bounded counter automata~\cite{Ibar78}.


\section{Preliminaries}\label{sec:preliminaries}
\subsection{Commutative Grammars.} 
Let $\Sigma=\{a_1,\ldots,a_m\}$ be a finite alphabet. The free monoid
generated by $\Sigma$ is denoted by $\Sigma^*$, and we denote by
$\Sigma^\odot$ the free commutative monoid generated by $\Sigma$. We
interchangeably use different equivalent ways in order to represent a
word $w\in \Sigma^\odot$. For $1\le j\le m$, let $i_j$ be the number
of times $a_j$ occurs in $w$, we equivalently write $w$ as
$w=a_1^{i_1}a_2^{i_2}\cdots a_m^{i_m}$, $w=(i_1,i_2,\ldots,i_m)\in
\N^m$ or $w:\Sigma \to \N$ with $w(a_j)=i_j$, whatever is most
convenient. By $\abs{w} = \sum_{1\le j\le m} i_j$ we the denote the
length of $w$, and the representation size $\# w$ of $w$ is
$\sum_{1\le j\le m} \ceil{\log i_j}$. Given $v,w\in \Sigma^\odot$, we
sometimes write $v+w$ in order to denote the concatenation $v\cdot w$
of $v$ and $w$. The empty word is denoted by $\epsilon$, and as usual
$\Sigma^+\defeq \Sigma^*\setminus \{\epsilon\}$ is the free semi-group
and $\Sigma^\oplus \defeq \Sigma^\odot\setminus \{ \epsilon \}$ the
free commutative semi-group generated by $\Sigma$. For $\Gamma
\subseteq \Sigma$, $\pi_\Gamma(w)$ denotes the projection of $w$ onto
alphabet symbols from $\Gamma$.

A commutative grammar (sometimes just grammar subsequently) is a tuple
$G=(N,\Sigma,S,P)$, where
\begin{itemize}
\item $N$ is the finite set of non-terminal symbols;
\item $\Sigma$ is a finite alphabet, the set of terminal symbols, such
  that $N\cap \Sigma=\emptyset$;
\item $S\in N$ is the axiom; and
\item $P\subset N^\oplus \times {(N\cup \Sigma)}^\odot$ is a finite
  set of productions.
\end{itemize}
The size of $G$, denoted by $\size{G}$, is defined as
\[
\size{G} \defeq \abs{N} + \abs{\Sigma} + \sum_{(V,W)\in P} \abs{V} + \abs{W}.
\]
Note that commutative words in $G$ are encoded in unary. Unless stated
otherwise, we use this definition of the size of a commutative grammar
in this paper.

Subsequently, we write $V\rightarrow W$ whenever $(V,W)\in P$. Let
$D,E\in {(N\cup \Sigma)}^\odot$, we say $D$ directly generates $E$,
written $D \Rightarrow_G E$, iff there are $F\in {(N\cup
  \Sigma)}^\odot$ and $V\rightarrow W \in P$ such that $D=V + F$ and
$E=F+W$.  We write $\Rightarrow_G^*$ to denote the reflexive
transitive closure of $\Rightarrow_G$, and if $U\Rightarrow^*_G V$ we
say that $U$ generates $V$. If $G$ is clear from the context, we omit
the subscript $G$. For $U\in N^\oplus$, the reachability set
$\reach(G,U)$ and the language $\lan(G,U)$ generated by $G$ starting
at $U$ are defined as
\begin{align*}
  \reach(G,U)  \defeq \{ W\in {(N\cup \Sigma)}^\odot : U \Rightarrow^* W \} &&
  \lan(G,U) \defeq \reach(G,U) \cap \Sigma^\odot.
\end{align*}
The reachability set $\reach(G)$ and the language $\lan(G)$ of $G$ are
then defined as $\reach(G)\defeq \reach(G,S)$ and $\lan(G)\defeq
\lan(G,S)$. The word problem is, given a commutative grammar $G$ and
$w\in \Sigma^\odot$, is $w\in \lan(G)$?  Our main focus in this paper
is, however, on the complexity of deciding language inclusion and
equivalence for commutative grammars: Given commutative grammars
$G,H$, language inclusion is to decide $\lan(G)\subseteq \lan(H)$, and
language equivalence is to decide $\lan(G) = \lan(H)$. Since our
grammars admit non-determinism, language inclusion and equivalence are
logarithmic-space inter-reducible.

By imposing restrictions on the set of productions, we obtain various
classes of commutative grammars. Following~\cite{Huy83,MW13}, given
$G=(N,\Sigma,S,P)$, we say that $G$ is
\begin{itemize}
\item of \emph{type-}0 if there are no restrictions on $P$;
\item \emph{context-sensitive} if $\abs{V}\ge \abs{W}$ for each
  $V\rightarrow W\in P$;
\item \emph{exponent-sensitive} if $V\in \{ {\{U\}}^\oplus : U\in N\}$
  for each $V\rightarrow W\in P$;
\item \emph{context-free} if $V\in N$ for each $V\rightarrow W\in P$;
\item \emph{regular} if $V\in N$ and $W \in (N\cup
  \{\epsilon\}) \cdot \Sigma^\odot$ for each $V\rightarrow W\in P$.
\end{itemize}

Equivalence problems for commutative grammars were studied by Huynh,
who showed that it is undecidable for context-sensitive and hence
type-0 grammars, and $\ComplexityFont{\Pi_2^P}$-hard and in
\coNEXP\ for regular and context-free commutative
grammars~\cite{Huy85}. The main contribution of this paper is to prove
the following theorem.

\begin{theorem}\label{thm:main}
  The language equivalence problem for regular and context-free
  commutative grammars problem is \coNEXP-complete.
\end{theorem}

Exponent-sensitive grammars were only recently introduced by Mayr and
Weihmann~\cite{MW13}. They showed that the word problem
is \PSPACE-hard, and that language equivalence is
\PSPACE-hard and in \ComplexityFont{2}-$\EXPSPACE$. The lower bounds
require commutative words on the left-hand sides of productions to be
encoded in binary. The second contribution of our paper is to improve
those results as follows.
\begin{theorem}\label{thm:main-expsens}
  The language equivalence problem for exponent-sensitive commutative
  grammars is \coNEXP-hard and in
  \ComplexityFont{co}-$\ComplexityFont{2}\NEXP$.
\end{theorem}

\begin{remark}\label{rmk:binary-encoding}
  For context-free commutative grammars, it is with no loss of
  generality possible to assume binary encoding of commutative words,
  which has, for instance, been remarked in~\cite{KT10}. For example,
  given a production $V\rightarrow a^{2^n}$, $n>0$, we can introduce
  fresh non-terminal symbols $V_1,\ldots, V_n$ and replace
  $V\rightarrow a^{2^n}$ by $V \rightarrow V_1V_1$, $V_{n}\rightarrow
  a$ and $V_i \rightarrow V_{i+1}V_{i+1}$ for every $1\le
  i<n$. Clearly, the grammar obtained by this procedure generates the
  same language and only results in a sub-quadratic blow-up of the
  size of the resulting grammar.
\end{remark}

\subsection{Presburger Arithmetic, Linear Diophantine Inequalities, and Semi-Linear
Sets.} Let $\vec{u}=(u_1,\ldots,u_m)$, $\vec{v}=(v_1,\ldots,v_m)\in
\Z^m$, the sum of $\vec{u}$ and $\vec{v}$ is defined component-wise,
i.e., $\vec{u}+\vec{v}=(u_1+v_1,\ldots,u_m+v_m)$. Given $u\in \Z$,
$\hat{\vec{u}}$ denotes the vector consisting of $u$ in every
component and any appropriate dimension. Let $1\le i\le j\le m$, we
define $\pi_{[i,j]}(\vec{u})\defeq (u_i,\ldots,u_j)$. By
$\norm{\vec{u}}_\infty$ we denote the maximum norm of $\vec{u}$, i.e.,
$\norm{\vec{u}}_\infty\defeq \max\{\abs{u_i} : 1\le i\le n \}$. Let
$M, N\subseteq \Z^m$ and $k\in \Z$, as usual $M+N$ is defined as $\{
\vec{m}+\vec{n} : \vec{m}\in M,~\vec{n}\in N \}$ and $k\cdot M \defeq
\{ k\cdot \vec{m} : \vec{m} \in M \}$. Moreover, $\norm{M}_\infty
\defeq \max\{\norm{\vec{z}}_\infty : \vec{z}\in M \}$.  The size
$\size\vec{u}$ of $\vec{u}$ is $\#\vec{u}\defeq \sum_{1\le i\le m}
\ceil{\log \abs{u_i}}$, i.e., numbers are encoded in binary, and the
size of $M$ is $\#M\defeq \sum_{\vec{u}\in M}\size \vec{u}$.  For an
$m\times n$ matrix $A$ consisting of elements $a_{ij}\in \Z$,
$\norm{A}_{1,\infty} \defeq \max\{ \sum_{1\le j\le n} \abs{a_{ij}} :
1\le i\le m \}$.

Presburger arithmetic is the is the first-order theory of the
structure ${\langle \mathbb{N},0,1,+,\ge \rangle}$. In this paper,
atomic formulas of Presburger arithmetic are linear Diophantine
inequalities of the form
\[
\sum_{1\le i\le n} a_i\cdot x_i \ge z_i,
\]
where $a_i, z_i\in \Z$ and the $x_i$ are first-order
variables. Formulas of Presburger arithmetic can then be obtained in
the usual way via positive Boolean combinations of atomic formulas and
existential and universal quantification over first-order variables,
i.e., according to the following grammar:
\[
\phi ::= \forall \vec{x}.\phi \mid \exists \vec{x}.\phi \mid \phi \wedge \phi
\mid \phi \vee \phi \mid t
\]
Here, the $\vec{x}$ range over tuples of first-order variables, and
$t$ ranges over linear Diophantine inequalities as above. We assume
that formulas of Presburger arithmetic are represented as a syntax
tree, with no sharing of sub-formulas.

Given a formula $\phi$ of Presburger arithmetic with no free
variables, validity is to decide whether $\phi$ holds with respect to
the standard interpretation in arithmetic.  By $\norm{\phi}_\infty$ we
denote the largest constant occurring in $\phi$, and $\abs{\phi}$ is
the length of $\phi$, i.e., the number of symbols required to write
down $\phi$, where constants are represented in unary. In analogy to
matrices, we define $\norm{\phi}_{1,\infty} \defeq \norm{\phi}_\infty
\cdot \abs{\phi}$. Let $\psi(\vec{x})$ be a quantifier-free formula
open in $\vec{x}=(x_1,\ldots,x_m)$ and
$\vec{x}^*=(x_1^*,\ldots,x_m^*)\in \N^m$, we denote by
$\psi[\vec{x}^*/\vec{x}]$ the formula obtained from $\psi$ by
replacing every $x_i$ in $\psi$ by $x_i^*$. Finally, given a
quantifier-free Presburger formula $\psi$ containing linear
Diophantine inequalities $t_1,\ldots,t_k$ and $b_1,\ldots,b_k\in
\{0,1\}$, $\psi[b_1/t_1,\ldots,b_k/t_k]$ denotes the Boolean formula
obtained from $\psi$ by replacing every $t_i$ with $b_i$.

In this paper, we are in particular interested in the $\Pi_2$-fragment
of Presburger arithmetic, for which the following is
known\footnote{The hardness result is stated under polynomial-time
  many-one reductions in~\cite{Haa14} .}.
\begin{proposition}[\cite{Gra89,Haa14}]\label{prop:pa-pi-2-complexity}
  Validity in the $\Pi_2$-fragment of Presburger arithmetic is
  \coNEXP-complete.
\end{proposition}

The sets of natural numbers definable in Presburger arithmetic are
semi-linear sets~\cite{GS64}. Let $\vec{b}\in \N^m$ and $P=\{
\vec{p}_1,\ldots,\vec{p}_n \}$ be a finite subset of $\N^m$, define
\[
\cone(P) \defeq \left\{ \lambda_1\cdot \vec{p}_1 + \cdots +
\lambda_n\cdot \vec{p}_n : \lambda_i\in \N,~1\le i\le n \right\}.
\]
A linear set $L(\vec{b},P)$ with base $\vec{b}$ and periods $P$ is
defined as $L(\vec{b}, P) \defeq \vec{b} + \cone(P)$. A semi-linear
set is a finite union of linear sets. For convenience, given a finite
subset $B$ of $\N^m$, we define
\[
L(B,P) \defeq \bigcup_{\vec{b}\in B} L(\vec{b}, P).
\]
The size of a semi-linear set $M = \bigcup_{i\in I}
L(B_i,P_i)\subseteq \N^m$ is defined as
\[
\size{M} \defeq \sum_{i\in I} \# B_i + \abs{B_i}\cdot \#P_i.
\]
In particular, numbers are encoded in binary. Given a semi-linear set
$N\subseteq \N^m$, $\#N$ is the minimum over the sizes of all
semi-linear sets $M=\bigcup_{i\in I}L(\vec{b}_i,P_i)$ such that $N=M$.

A system of linear Diophantine inequalities $D$ is a conjunction of
linear inequalities over the same first-order variables
$\vec{x}=(x_1,\ldots,x_n)$, which we write in the standard way as $D:
A\cdot \vec{x} \ge \vec{c}$, where $A$ is a $m\times n$ integer matrix
and $\vec{c}\in \N^m$. The size $\#D$ of $D$ is the number of symbols
required to write down $D$, where we assume binary encoding of
numbers. The set of solutions of $D$ is denoted by $\eval{D}\subseteq
\N^n$. We say that $D$ is feasible if $\eval{D}\neq \emptyset$.
In~\cite{Pot91,Dom91}, bounds on the semi-linear representation of
$\eval{D}$ are established. The following proposition is a consequence
of Corollary~1 in~\cite{Pot91} and Theorem~5 in~\cite{Dom91}.
\begin{proposition}[\cite{Pot91,Dom91}]\label{prop:pottier}
  Let $D : A\cdot \vec{x} \ge \vec{c}$ be a system of linear
  Diophantine inequalities such that $A$ is an $m\times n$
  matrix. Then $\eval{D}=L(B,P)$ for $B,P\subseteq \N^n$ such that
  $\abs{P}\le \binom{n}{m}$ and 
  \[
  \norm{B}_\infty,\norm{P}_\infty \le 
  {(\norm{A}_{1,\infty} + \norm{\vec{c}}_\infty + 2)}^{m+n}.
  \]
\end{proposition}




\section{Lower Bounds}
In this section, we establish the \coNEXP-lower bound of
Theorems~\ref{thm:main} and~\ref{thm:main-expsens}. This is shown by
establishing \coNEXP-hardness of language inclusion for regular
commutative grammars. However, for the sake of a clear presentation,
we will first describe the reduction for context-free commutative
grammars, and then show how the approach can be adapted to regular
commutative grammars. Finally, we also show that even reachability
equivalence is \coNEXP-hard.

As stated in the introduction, we reduce from validity in the
$\Pi_2$-fragment of Presburger arithmetic.  To this end, let $\phi =
\forall \vec{x}. \exists \vec{y}.\psi(\vec{x},\vec{y})$ such that
$\vec{x} = (x_1,\ldots,x_m)$, $\vec{y} = (y_1,\ldots,y_n)$, and $\psi$
is a positive Boolean combination of atomic formulas
$t_1,\ldots,t_k$. For our reduction, we write atomic formulas of
$\psi$ as
\[
t_i : \sum_{1\le j\le m} (a_{i,j}^+ - a_{i,j}^-) \cdot x_j + z_i^+ - z_i^- 
\ge \sum_{1\le j\le n} (b_{i,j}^+ - b_{i,j}^-) \cdot y_j,
\]
where the $a_{i,j}^+,a_{i,j}^-\in \N$ are such that $a_{i,j}^+=0$ or
$a_{i,j}^-=0$, and likewise the $b_{i,j}^+,b_{i,j}^-\in \N$ are such
that $b_{i,j}^+=0$ or $b_{i,j}^-=0$, and the $z_i^+,z_i^-\in \N$ such
that $z_i^+=0$ or $z_i^-=0$. Moreover, in the following we set
$a_{i,j}\defeq a_{i,j}^+ - a_{i,j}^-$, $b_{i,j}\defeq b_{i,j}^+ -
b_{i,j}^-$ and $z_i\defeq z_i^+ - z_i^-$.

\begin{example}\label{ex:running-ex-1}
  Consider the formula $\phi$ that we will use as our running example
  such that $\phi = \forall x.\exists y.\psi(x,y)$, where
  $\psi(x,y)=(t_1 \wedge t_2) \vee (t_3 \wedge t_4)$ and
  \begin{align*}
    t_1 & = x \ge 2\cdot y & t_3 & = x + 1\ge 2\cdot y\\
    t_2 & = -x \ge -2\cdot y & t_4 & =  -x - 1 \ge -2 \cdot y,
  \end{align*}
  which expresses that every natural number is either even or
  odd. Here, for instance, $a_{2,1}^+ = 0$, $a_{2,1}^- = 1$, $z_1^+=
  z_1^-=0$, $b_{2,1}^+=0$ and $b_{2,1}^-=2$. Hence $a_{2,1} = -1$,
  $z_2=0$ and $b_{2,1}=-2$.\qedx
\end{example}

With no loss of generality and due to unary encoding of numbers in
$\phi$, we may assume that the following inequalities hold:
\begin{align}\label{eqn:formula-inequalities}
  \abs{\phi} & \ge 2 + m + n + k & \abs{\phi} & \ge \norm{\phi}_\infty
\end{align}
We furthermore define the constant $c\in \N$, whose bit representation
is polynomial in $\abs{\phi}$, as
\begin{align}\label{eqn:constant-c}
  c \defeq \min \left\{ 2^n \ge 
  \abs{\phi}^{3\cdot \abs{\phi} + 2}\cdot 2^{|\phi|} : n\in \N \right\}.
\end{align}

Let $\Sigma\defeq \{ t_1^+, t_1^-,\ldots, t_k^+,t_k^- \}$, we now show
how to construct in logarithmic space context-free commutative
grammars $G,H$ over $\Sigma$ such that $\lan(G)\subseteq \lan(H)$ iff
$\phi$ is valid. The underlying idea is as follows: the language of
$G$ consists of all possible values of the left-hand sides of the
inequalities $t_i$ for every choice of $\vec{x}$, where the value of
some $t_i$ is represented by a word $w\in \Sigma^\odot$ via the
difference $w(t_i^+)-w(t_i^-)$. For every $w\in \Sigma^\odot$ and
$1\le i\le k$, we misuse notation and define $w(t_i)\defeq
w(t_i^+)-w(t_i^-)\in \Z$; note that in particular $t_i\not\in
\Sigma$. The grammar $H$ can then be defined in an analogous way and
produces the values of the right-hand sides of $H$ for a choice of
$\vec{y}$, but can in addition simulate the Boolean structure of
$\psi$ in order to tweak those $t_i$ for which, informally speaking,
it cannot obtain a good value. We explain the reduction in further
detail in due course, but for now give a small example and then turn
towards the formal definition of $G$.
\begin{example}\label{ex:running-ex-2}
  Let $\phi$ be our running example. We have that $\Sigma = \{
  t_1^+,t_1^-,\ldots,t_4^+,t_4^- \}$. Our goal is to define a
  context-free commutative grammar $G$ such that,
  \[
  \lan(G) = \{ (c,c,c,c,c+1,c,c,c+1) + i\cdot 
  (c+1,c,c,c+1,c+1,c,c,c+1)\in \Sigma^\odot : i\in \N \}.
  \]
  For any $w\in \lan(G)$, there is then some $x\in \N$ such that
  $w(t_1)=x$, $w(t_2)=-x$, $w(t_3)=x+1$ and $w(t_4)=-x-1$, those
  values which correspond to the possible values of the left-hand
  sides of the atomic formulas of $\phi$ for any choice of $x$.\qedx
\end{example}

Recall that we may represent commutative words of $\Sigma^\odot$ as
vectors of natural numbers, we define:
\begin{align}
  \label{eqn:u-def} u & \defeq (z_1^+,z_1^-,\ldots,z_k^+,z_k^-)\in \Sigma^\odot\\
  \label{eqn:v-i-def} v_i & \defeq (a_{1,i}^+,a_{1,i}^-, \ldots,a_{k,i}^+,a_{k,i}^-)
  \in \Sigma^\odot
  & (1\le i\le m)
\end{align}
The grammar $G$ is constructed as $G\defeq (N_G,\Sigma,S_G,P_G)$,
where $N_G \defeq \{ S, X\}$ and $P_G$ is defined as follows:
\begin{align*}
  S_G   & \rightarrow X \hat{\vec{c}} u  &   
  X & \rightarrow X \hat{\vec{c}} v_i  & (1\le i\le m)\\
  X   & \rightarrow \epsilon
\end{align*}
Here, $c$ is the constant from~(\ref{eqn:constant-c}) whose addition
ensures that the values of the $t_i^+$ and $t_i^-$ generated by $G$
are large. Moreover, recall that it follows from
Remark~\ref{rmk:binary-encoding} that $G$ can be constructed in
logarithmic space even though $c$ is exponential in $\abs{\phi}$. The
following lemma captures the essential properties of $G$.

\begin{lemma}\label{lem:lhs-language}
  Let $G$ be as above. The following hold:
  \begin{enumerate}[(i)]
  \item For every $\vec{x}\in \N^m$ there exists $w\in \lan(G)$ such
    that for all $1\le i\le k$,
    \[
      w(t_i) = 
      \sum_{1\le j\le m} (a_{i,j}^+ - a_{i,j}^-) \cdot x_j + z_i^+ - z_i^-.
    \]
  \item For every $w\in \lan(G)$ there exists $\vec{x}\in \N^m$ such
    that for all $1\le i\le k$,
    \begin{align}
      \label{eqn:ineq-by-x-1}
      w(t_i) & = 
      \sum_{1\le j\le m} (a_{i,j}^+ - a_{i,j}^-) \cdot x_j + z_i^+ - z_i^-\\
      \label{eqn:ineq-by-x-2} w(t_i^+) & \ge c+ z_i^+ + \sum_{1\leq j \leq m} 
      c \cdot x_j \ge c\cdot (1 + \norm{\vec{x}}_\infty)\\
      \label{eqn:ineq-by-x-3} w(t_i^-) & \ge c+ z_i^- + \sum_{1\leq j \leq m} 
      c \cdot x_j \ge c \cdot (1 + \norm{\vec{x}}_\infty).
    \end{align}
  \end{enumerate}
\end{lemma}
\begin{proof}
  Regarding~(i), let $\vec{x}=(x_1,\ldots,x_m)\in \N^m$ and consider
  the following derivation of $G$:
  \[
  S_G \Rightarrow X u \Rightarrow X v_1 u \Rightarrow^*
  X  v_1^{x_m} u \Rightarrow^* X  v_1^{x_1} \cdots v_{m-1}^{x_{m-1}} u
  \Rightarrow^* v_1^{x_1}\cdots v_{m-1}^{x_{m-1}} v_m^{x_m} u = w.
  \]
  For every $1\le i\le k$ we have
  \begin{align*}
  w(t_i^+) & = v_1(t_i^+) \cdot x_1 + \cdots + v_m(t_i^+) \cdot x_m + u(t_i^+)\\
  & = (a_{i,1}^+ + c)\cdot x_1 + \cdots + (a_{i,m}^++ 
  c)\cdot x_1 + z_i^+ + c\\
  & = \sum_{1\le j\le m} (a_{i,j}^+ + c)\cdot x_j + z_i^+ + c.
  \end{align*}
  In the same way, we obtain
  \[
  w(t_i^-) = \sum_{1\le j\le m} (a_{i,j}^- + c)\cdot x_j + z_i^- + c,
  \]
  whence
  \[
  w(t_i) = w(t_i^+) - w(t_i^-) = 
  \sum_{1\le j\le m} (a_{i,j}^+ - a_{i,j}^-) \cdot x_j + z_i^+ - z_i^-.
  \]
  
  Regarding~(ii), by the construction of $G$, any $w\in \lan(G)$ is
  of the form
  \[
  w = v_1^{x_1}\cdots v_{m-1}^{x_{m-1}} v_m^{x_m} u.
  \]
  Define $\vec{x}\defeq (x_1,\ldots,x_m)$ and let $1\le i\le k$,
  Equation~(\ref{eqn:ineq-by-x-1}) follows as in~(i). Moreover, since
  $u(t_i^+)=c+z_i^+$ and $v_j(t_i^+)\ge c$ for every $1\le j\le m$, we
  obtain inequality~(\ref{eqn:ineq-by-x-2}). The same argument allows
  for deriving~(\ref{eqn:ineq-by-x-3}).
\end{proof}
This completes the construction of $G$ and the proof of the relevant
properties of $G$. We now turn towards the construction of $H\defeq
(N_H,\Sigma,S_H,P_H)$ and define the set of non-terminals $N_H$ and
productions $P_H$ of $H$ in a step-wise fashion. Starting in $S_H$,
$H$ branches into three gadgets starting at the non-terminal symbols
$Y$, $F_\psi$ and $I$:
\[
S_H \rightarrow Y F_\psi I
\]
Here, $Y$ is an analogue to $X$ in $G$. Informally speaking, it allows
for obtaining the right-hand sides of the inequalities $t_i$ for a
choice of $\vec{y}\in \N^n$. In analogy to $G$, we define
\begin{align*}
  w_i & \defeq (b_{1,i}^+, b_{1,i}^-, \ldots, b_{k,i}^+,b_{k,i}^-)\in \Sigma^\odot 
  & (1\le i \le n)\\
  Y & \rightarrow Y w_i & (1\le i\le n)\\
  Y & \rightarrow \epsilon
\end{align*}
In contrast to $X$ from $G$, note that $Y$ does not add $\vec{c}$
every time it loops. The following lemma is the analogue of $H$ to
Lemma~\ref{lem:lhs-language} and can be shown along the same lines.
\begin{lemma}{\label{lem:rhs-y-language}}
  Let $Y$ be the non-terminal of $H$ as defined above. The following
  hold:
  \begin{enumerate}[(i)]
  \item For every $\vec{y}\in \N^n$ there exists $w\in \lan(H,Y)$ such that
    for all $1\le i \le k$, $w(t_i^+) = \sum_{1\le j\le n} b_{i,j}^+
    \cdot y_j$, $w(t_i^-) = \sum_{1\le j\le n} b_{i,j}^- \cdot y_j$, and
\[
      w(t_i) = \sum_{1\le j\le n} (b_{i,j}^+ - b_{i,j}^-) \cdot y_j.
\]
  \item For every $w\in \lan(H,Y)$ there exists $\vec{y}\in \N^n$ such
    that for all $1\le i \le k$,
\[
      w(t_i) = \sum_{1\le j\le n} (b_{i,j}^+ - b_{i,j}^-) \cdot y_j.
\]   
  \end{enumerate}
\end{lemma}
It is clear that the $w_Y$ generated by $Y$ may not be able to
generate all $t_i$ in a way that match all $w$ generated by $G$ (i.e.,
all choices of $\vec{x}$ made through $G$). For now, let us assume
that $w(t_i^+)\ge w_Y(t_i^+)$ and $w(t_i^-)\ge w_Y(t_i^-)$ holds for
every $1\le i\le k$. Later, we will show that if there is a good
choice for $\vec{y}$, we can find a good $w_Y\in \lan(H,Y)$ with this
property. After generating $w_Y$, informally speaking, $H$ should
produce $t_i^+$ and $t_i^-$ in order match $w$, provided that $\psi$
is valid.

In particular, the Boolean structure of $\psi$ enables us to produce
arbitrary quantities of some $t_i$.
This is the duty of the gadget $F_\psi$ which allows for assigning
arbitrary values to some atomic formulas $t_i$ via gadgets $R_{t_i}$
defined below. The gadget $F_\psi$ recursively traverses the matrix
formula $\psi$ and invokes some $R_\gamma$ whenever a disjunction is
processed and a disjunct $\gamma$ is evaluated to false:
\begin{align*}
  F_{t_i} & \rightarrow \epsilon & F_{\alpha \wedge \beta} & \rightarrow F_\alpha F_\beta\\
  F_{\alpha \vee \beta} & \rightarrow F_\alpha R_\beta & 
  F_{\alpha \vee \beta} & \rightarrow R_\alpha F_\beta\\
  F_{\alpha \vee \beta} & \rightarrow F_\alpha F_\beta
\end{align*}
The definition of $R_\gamma$ for every subformula $\gamma$ of $\psi$
occurring in the syntax tree of $\psi$ is now not difficult: we
traverse $\gamma$ until we reach a leaf $t_i$ of the syntax tree of
$\gamma$ and then allow for generating an arbitrary number of alphabet
symbols $t_i^+$ and $t_i^-$. Let $1\le i\le k$, we define the
following productions:
\begin{align*}
  R_{t_i} & \rightarrow \epsilon & R_{t_i} & \rightarrow R_{t_i}t_i^+
  & R_{t_i} & \rightarrow R_{t_i}t_i^-\\
  R_{\alpha \wedge \beta} & \rightarrow R_\alpha R_\beta &
  R_{\alpha \vee \beta} & \rightarrow R_\alpha R_\beta
\end{align*}

\begin{example}
  Continuing our running example, $H$ includes among others the
  following rules:
  \begin{align*}
    F_{(t_1 \wedge t_2) \vee (t_3 \wedge t_4)} & \rightarrow 
    F_{(t_1 \wedge t_2)} R_{(t_3 \wedge t_4)} &
    F_{(t_1 \wedge t_2) \vee (t_3 \wedge t_4)} & \rightarrow 
    R_{(t_1 \wedge t_2)} F_{(t_3 \wedge t_4)}\\
    F_{(t_1 \wedge t_2) \vee (t_3 \wedge t_4)} & \rightarrow 
    F_{(t_1 \wedge t_2)} F_{(t_3 \wedge t_4)} &
    R_{(t_1\wedge t_2)} & \rightarrow R_{t_1}R_{t_2} \\
    R_{t_1} & \rightarrow R_{t_1}t_1^+ & R_{t_2} & \rightarrow R_{t_2}t_2^+\\
    R_{t_1} & \rightarrow R_{t_1}t_1^-& R_{t_2} & \rightarrow R_{t_2}t_2^-\\
    R_{t_1} & \rightarrow \epsilon & R_{t_2} & \rightarrow \epsilon
  \end{align*}
  In particular,
  \begin{multline*}
    \lan(H,F_\psi) = \{ (n_1^+,n_1^-,n_2^+,n_2^-,0,0,0,0)
    \in \Sigma^\odot : n_1^+,n_1^-,n_2^+,n_2^-\in \N
    \}\cup \\
    \cup \{ (0,0,0,0,n_3^+,n_3^-,n_4^+,n_4^-)
    \in \Sigma^\odot : n_3^+,n_3^-,n_4^+,n_4^-\in \N
    \},
    \end{multline*}
  i.e., $F_\psi$ may produce arbitrary quantities of either $t_1$ and
  $t_2$, or $t_3$ and $t_4$, reflecting the Boolean structure of
  $\psi$.\qedx
\end{example}

Finally, it remains to provide a possibility to increase $w_Y(t_i)$
for those $t_i$ that were not processed by some $R_{t_i}$ in order to
match $w$.  For a good choice of $w_Y$, we certainly have $w(t_i)\ge
w_Y(t_i)$ for those $t_i$. Hence, in order to make $w_Y$ agree with
$w$ on $t_i$, all we have to do to $w_Y$ is to non-deterministically
increment, i.e., produce, $t_i^+$ at least as often as $t_i^-$. This
is the task of the gadget $I$ of $H$, whose production rules are as
follows:
\begin{align*}
  I & \rightarrow \epsilon & I & \rightarrow I t_i^+ t_i^- & I &
  \rightarrow It_i^+ & (1\le i\le k)
\end{align*}
The subsequent lemma, whose proof is immediate, states the properties
of $I$ formally.
\begin{lemma}\label{lem:rhs-i-language}
  $\lan(H,I) = \left\{ (n_1^+,n_1^-,\ldots,n_k^+,n_k^-)\in \Sigma^\odot :
  n_j^+ \ge n_j^-,~ 1\le j\le k \right\}$.
\end{lemma}
This completes the construction of $H$. Before we prove two lemmas
that establish the correctness of our construction, let us illustrate
the necessity of $I$ with the help of an example.
\begin{example}
  For $\phi$ of our running example, suppose that $x$ is even and hence
  \begin{multline*}
  w = (c+(1+c)\cdot x,\ c +c\cdot x,\ c+c\cdot x,\ c+(1+c)\cdot x,\\   
  c+1+ (1+c)\cdot x,\ c+c\cdot x,\ c+c\cdot x,\ c+1+ (1+c)\cdot x)\in \lan(G).
  \end{multline*}
  For $y=x/2$ , we clearly have
  \[
  w_Y = (2\cdot y ,0,0,2\cdot y,2\cdot y,0,0,2\cdot y)\in \lan(H,Y).
  \]
  The value of $w_Y(t_1)$ and $w_Y(t_2)$ already matches $w(t_1)$ and
  $w(t_2)$, respectively, however $w_Y(t_1^+)\neq w(t_1^+)$, etc. But
  then we find
  \[
  w_I = (2\cdot y+1) \cdot (c,c,c,c,0,0,0,0) \in \lan(H,I).
  \]
  Moreover, we have
  \[
  w_F = (0,0,0,0,1,0,0,1) +  (2\cdot y+1)  \cdot (0,0,0,0,c,c,c,c)
  \in \lan(H,F_\psi),
  \]
  and consequently $w_Y + w_F + w_I = w \in \lan(H)$.\qedx
\end{example}
\begin{lemma}\label{lem:inclusion-valid}
  Suppose $\lan(G)\subseteq \lan(H)$, then $\phi = \forall
  \vec{x}.\exists \vec{y}.\psi(\vec{x},\vec{y})$ is valid.
\end{lemma}
\begin{proof}
  Let $\vec{x}\in \N^m$, we show how to construct $\vec{y}\in \N^n$
  such that $\psi(\vec{x}, \vec{y})$ evaluates to true. By
  Lemma~\ref{lem:lhs-language}(i), there exists $w\in \lan(G)$ such
  that for all $1\le i\le k$,
  \begin{align}\label{eqn:ineq-by-x}
    w(t_i) =
    \sum_{1\le j\le m} (a_{i,j}^+ - a_{i,j}^-) \cdot x_j + z_i^+ - z_i^-,
  \end{align}
  and since $\lan(G)\subseteq \lan(H)$, $w\in \lan(H)$. By definition
  of $H$, there are $w_Y, w_F, w_I\in \Sigma^\odot$ such that
  \begin{itemize}
  \item $w_Y\in \lan(H,Y)$,
    $w_I\in \lan(H,I)$, $w_F\in \lan(H,F_\psi)$; and
  \item $w = w_Y + w_F + w_I$.
  \end{itemize}
  By Lemma~\ref{lem:rhs-y-language}(ii), there exists $\vec{y}\in
  \N^n$ such that for all $1\le i\le k$,
  \begin{align}\label{eqn:ineq-by-y}
    w_Y(t_i) = \sum_{1\le j\le n} (b_{i,j}^+ - b_{i,j}^-) \cdot y_j.
  \end{align}
  We claim that $\vec{y}$ has the desired properties, i.e., that
  $w(t_i)\ge w_Y(t_i)$ for all inequalities necessary to make
  $\psi(\vec{x},\vec{y})$ evaluate to true. To this end, consider the
  derivation tree of $w_F$ and define a mapping $\xi: \{1,\ldots, k\}
  \to \{0, 1\}$ such that $\xi(i)\defeq 0$ if the non-terminal
  $R_{t_i}$ occurs in the derivation tree and $\xi(i)\defeq 1$ if
  $F_{t_i}$ occurs in it. By the construction of $F_\psi$, this
  mapping is well-defined, and moreover it also implies that
  $\psi[\xi(t_1)/t_1,\ldots, \xi(t_k)/t_k]$ evaluates to true. So it
  remains to show that for all $t_i$ such that $\xi(i)=1$, i.e., we
  have
  \[
  w(t_i) = w(t_i^+) - w(t_i^-) \ge w_Y(t_i^+) - w_Y(t_i^-) = w_Y(t_i).
  \]
  Since for all such $i$ we have $w_F(t_i^+)=w_F(t_j^-)=0$, it follows
  that $w(t_i^+)= w_Y(t_i^+) + w_I(t_i^+)$ and $w(t_i^-)= w_Y(t_i^-) +
  w_I(t_i^-)$, hence
\[
  w(t_i^+) - w(t_i^-) = (w_Y(t_i^+) - w_Y(t_i^-)) + (w_I(t_i^+) - w_I(t_i^-)).
\]
  By Lemma~\ref{lem:rhs-i-language}, $ w_I(t_i^+) - w_I(t_i^-) \ge 0$,
  and consequently $w(t_i) \ge w_Y(t_i)$ as required
  by~(\ref{eqn:ineq-by-x}) and~(\ref{eqn:ineq-by-y}).
\end{proof}
The converse direction is slightly more involved. Informally speaking,
on the first sight one might be worried that $H$ produces more $t_i^+$
or $t_i^-$ than $G$, which cannot be ``erased.'' However, the addition
of $c$ in every component for every reduction step made by $G$
together with Proposition~\ref{prop:pottier} allows us to overcome
this obstacle.
\begin{lemma}\label{lem:valid-inclusion}
  Suppose $\phi = \forall \vec{x}.\exists \vec{y}.
  \psi(\vec{x},\vec{y})$ is valid, then $\lan(G)\subseteq \lan(H)$.
\end{lemma}
\begin{proof}
  Let $w\in \lan(G)$, by Lemma~\ref{lem:lhs-language}(ii) there exists
  $\vec{x}^*\in \N^m$ such that (\ref{eqn:ineq-by-x-1}),
  (\ref{eqn:ineq-by-x-2}) and~(\ref{eqn:ineq-by-x-3}) hold. By
  assumption, there is $\vec{y}^*\in \N^n$ such that
  $\psi(\vec{x}^*,\vec{y}^*)$ holds. Hence, there is
  $\xi:\{1,\ldots,k\} \to \{0,1\}$ such that for all $i$ where
  $\xi(i)=1$,
%
  \[
  \sum_{1\le j\le m} a_{i,j} \cdot x_j^* + z_i
  \ge \sum_{1\le j\le n} b_{i,j} \cdot y_j^*
  \]
  and $\psi[\xi(1)/t_1, \ldots, \xi(k)/t_k]$ evaluates to true. With
  no loss of generality, write $\{ i: \xi(i) = 1\} = \{1,\ldots, h\}$
  for some $1\le h\le k$. Consider the system $D : A \cdot
  (\vec{x},\vec{y})\ge \vec{z}$ of linear Diophantine inequalities
  over the unknowns $\vec{x}$ and $\vec{y}$, where
  \begin{align*}
    A & \defeq \begin{pmatrix}
      a_{1,1} & \cdots & a_{1,m} & -b_{1,1} & \cdots & -b_{1,n}\\
      \vdots & \ddots & \vdots & \vdots & \ddots & \vdots\\
      a_{h,1} & \cdots & a_{h,m} & -b_{h,1} & \cdots & -b_{h,n}\\
    \end{pmatrix} & 
    \vec{z} & \defeq   
    \begin{pmatrix}
      -z_1\\
      \vdots\\
      -z_h
    \end{pmatrix}.
  \end{align*}
  By assumption, $D$ has a non-empty solution set. We have that $A$ is
  a $h\times (m+n)$ matrix with $\norm{M}_{1,\infty}\le
  \norm{\psi}_{1,\infty}$ and $\norm{\vec{z}}_\infty \le
  \norm{\psi}_{\infty}$. By Proposition~\ref{prop:pottier}, there are
  $B,P\subseteq \N^{m+n}$ such that $\eval{D}= B +
  \cone(P)$. Consequently,
\[
  \vec{x}^* = \proj_{[1,m]}(\vec{b} + \lambda_1 \cdot \vec{p}_1 + 
  \cdots + \lambda_\ell \cdot \vec{p}_\ell)
\]
  for some $\vec{b}\in B$, $\vec{p}_i\in P$ and $\lambda_i\in \N$. In
  particular, since $\abs{P}\le \binom{m+n}{h}\le 2^{\abs{\phi}}$ we
  have
  \begin{align}\label{eqn:lambda-inequality}
    0\le \sum_{1\le i\le \ell}\lambda_i\le
    \norm{\vec{x}^*}_\infty\cdot \ell\le \norm{\vec{x}^*}_\infty \cdot
    2^{\abs{\phi}}.
  \end{align}
  Now let
  \[
  \vec{y}^\dagger \defeq \proj_{[m+1,m+n]}(\vec{b} + \lambda_1 \cdot \vec{p}_1 + 
  \cdots + \lambda_\ell \cdot \vec{p}_\ell).
  \]
  We have $(\vec{x}^*,\vec{y}^\dagger)$ is a solution of $D$ and
  henceforth $\psi[\vec{x}^*/\vec{x}, \vec{y}^\dagger/\vec{y}]$
  evaluates to true. Moreover,
  \begin{align*}
    \norm{\vec{y}^\dagger}_\infty & \le \norm{\vec{b}}_\infty +
    \norm{\lambda_1 \cdot \vec{p}_1 + \cdots + \lambda_\ell \cdot
      \vec{p}_\ell}_\infty\\
    & \le \norm{B}_\infty + \sum_{1\le i\le \ell} \lambda_i \cdot 
    \norm{P}_\infty  \\
    & \le \norm{B}_\infty + \norm{\vec{x}^*}_{\infty} \cdot 2^{\abs{\phi}}\cdot \norm{P}_\infty
    & (\text{by}~(\ref{eqn:lambda-inequality})) \\
    & \le \left(1 + 
    \norm{\vec{x}^*}_\infty\cdot 2^{\abs{\phi}}\right) \cdot {(\norm{A}_{1,\infty} + \norm{\vec{z}}_\infty + 2)}^{h+m+n} & (\text{by Prop.}~\ref{prop:pottier})\\
    & \le \left(1 + \norm{\vec{x}^*}_\infty\cdot 2^{\abs{\phi}}\right) \cdot {((m+n+1)\cdot \norm{\phi}_\infty + 2)}^{k+m+n}\\
    & \le \left(1 + \norm{\vec{x}^*}_\infty\cdot 2^{\abs{\phi}}\right) \cdot \abs{\phi}^{3\cdot \abs{\phi}}\\
    & \le \left(1 + \norm{\vec{x}^*}_\infty\right)\cdot  \abs{\phi}^{3\cdot \abs{\phi}}\cdot 2^{\abs{\phi}}
    & (\text{by}~(\ref{eqn:formula-inequalities}))\\
    & \le (1 + \norm{\vec{x}^*}_\infty) \cdot \frac{c}{\abs{\phi}^2} 
    & (\text{by}~(\ref{eqn:constant-c}))
  \end{align*}
  Combining the estimation of $\norm{\vec{y}^\dagger}_\infty$
  with~(\ref{eqn:ineq-by-x-2}) and~(\ref{eqn:ineq-by-x-3}) of
  Lemma~\ref{lem:lhs-language}, for every $1\le i\le k$ we obtain
  \begin{align}\label{eqn:w-ti-inequality}
  w(t_i^+), w(t_i^-) \ge c \cdot (1 + \norm{\vec{x}^*}_\infty) \ge
  \norm{\vec{y}^\dagger}_\infty \cdot \abs{\phi}^2 \ge
  \norm{\vec{y}^\dagger}_\infty \cdot \norm{\phi}_\infty \cdot
  \abs{\phi}.
  \end{align}
  By Lemma~\ref{lem:rhs-y-language}(i) there is $w_Y\in \lan(H,Y)$
  such that~(\ref{eqn:w-ti-inequality}) yields
  \begin{align*}
    w(t_i^+) & \ge \sum_{1\le j\le n}  \norm{\vec{y}^\dagger}_\infty \cdot
    \norm{\phi}_\infty 
    \ge \sum_{1\le j\le n}b_{i,j}^+ \cdot y_j^\dagger  = w_Y(t_i^+)\\
    w(t_i^-) & \ge \sum_{1\le j\le n}
    \norm{\vec{y}^\dagger}_\infty \cdot  \norm{\phi}_\infty 
    \ge \sum_{1\le j\le n}b_{i,j}^- \cdot y_j^\dagger = w_Y(t_i^-).
  \end{align*}
  Moreover, the construction of $F_\psi$ is such that
  \[
  \left\{ w_F \in \Sigma^\odot : w_F(t_i^+)=w_F(t_i^-)=0,~\xi(i) = 1,~1\le i\le k 
  \right\} \subseteq \lan(H,F_\psi).
  \]
  Hence, we can find some $w_F\in \lan(H,F_\psi)$ which allows us to
  adjust those $t_i$ for which $\xi(i)=0$. More formally, for $1\le
  i\le k$ such that $\xi(i)=0$,
  \[
  (w_Y+w_F)(t_i^+)=w(t_i^+) \text{ and }
  (w_Y+w_F)(t_i^-)=w(t_i^-).
  \]
  On the hand, for all $1\le i\le k$ such that $\xi(i)=1$,
  \[
  (w_Y+w_F)(t_i^+)=w_Y(t_i^+) \text{ and }(w_Y+w_F)(t_i^+)=w_Y(t_i^+),
  \]
  i.e., those $t_i$ remain untouched by $w_F$.

  Consequently, it remains to show that there is a suitable $w_I\in
  \lan(H,I)$ such that we can adjust those $t_i$ which were left
  untouched by $w_F$ above.
  For all $1\le i\le k$ such that $\xi(i)=1$, since $\vec{y}^\dagger$
  is a solution of $D$, we have
  \begin{align*}
         & w(t_i) = w(t_i^+) - w(t_i^-) \ge w_Y(t_i^+) - w_Y(t_i^-) = w_Y(t_i)\\
    \iff & w(t_i^+) - w_Y(t_i^+) \ge w(t_i^-) - w_Y(t_i^-)\\
    \iff & \text{there are }m_i,n_i\in \N \text{ such that } w(t_i^+) = w_Y(t_i^+)
    + m_i + n_i \text{ and }\\ & w(t_i^-) = w_Y(t_i^-) + m_i.
  \end{align*}
  But then Lemma~\ref{lem:rhs-i-language} yields the required
  $w_{I}\in \lan(H,I)$ such that $w_{I}(t_i^+)=m_i+n_i$,
  $w_{I}(t_i^-)=m_i$, and $w_{I}(t_j^+)=w_{I}(t_j^+)=0$ for all $j$
  such that $\xi(j)=0$.

  Summing up, we have $w=w_Y+w_I+w_F$, and hence $w\in \lan(H)$ as
  required.
\end{proof}

Lemmas~\ref{lem:inclusion-valid} and~\ref{lem:valid-inclusion}
together with Proposition~\ref{prop:pa-pi-2-complexity} yield the
\coNEXP-lower bound of Theorems~\ref{thm:main}
and~\ref{thm:main-expsens} of the language inclusion problem for
context-free and exponent-sensitive grammars. In order to show
hardness of the equivalence problem, we merge $H$ into $G$, i.e.,
define
\begin{align*}
G^e\defeq (N_G\cup N_H \cup \{S\}, \Sigma, S, P_G\cup P_H \cup \{ S \rightarrow
S_G, S\rightarrow S_H \}).
\end{align*}
It is now clear that $\phi$ is valid iff $\lan(G^e)=\lan(H)$. Finally,
if we, in addition, redefine $H$ as
\begin{align*}
  H^e \defeq (\{S_G, X\} \cup N_H \cup \{ S\}, \Sigma, P_H \cup \{ S\rightarrow S_G, 
  S\rightarrow XS_H, X\rightarrow \epsilon\})
\end{align*}
then $G^e$ and $H^e$ have the same set of non-terminals $N\defeq
N_G\cup N_H \cup \{S\} = \{ S_G, X \} \cup N_H \cup \{ S\}$, and even
a stronger statement holds:
\begin{align}\label{eqn:bpp-reduction}
  \phi~\text{ is valid} \iff \reach(G^e) = \reach(H^e).
\end{align}
%

\subsection{Hardness for Regular Commutative Grammars}
It remains to show how the reduction developed so far can be adapted
in order to prove \coNEXP-hardness of the equivalence problem for
regular commutative grammars. As constructed above, neither $G$ nor
$H$ are regular. In this section, we show how to obtain regular
commutative grammars $G^r$ and $H^r$ from $G$ and $H$ such that
$\lan(G^r)\subseteq \lan(H^r)$ iff $\lan(G) \subseteq \lan(H)$.

It is actually not difficult to see that $H$ can be made regular. By
the construction of $H$, both gadgets starting in $Y$ and $I$ are
regular, but $F_\psi$ is not, and also the initial production
$S_H\rightarrow YF_\psi I$ is not regular. The latter can be fixed by
additionally adding productions $Y\rightarrow F_\psi$ and
$F_\psi\rightarrow I$ to the set of productions $P$, and replacing $S_H\rightarrow YF_\psi I$
with $S_H\rightarrow Y$.  It thus remains to make $F_\psi$
regular. The non-regularity of the latter is due to the fact that we
use branching provided by context-free commutative grammars in order
to simulate the Boolean structure of $\psi$. However, as we show now,
it is possible to serialise $F_\psi$.

As a first step, we discuss the serialisation of $R_\gamma$ for all
subformulas $\gamma$ occurring in the syntax-tree of $\psi$. Recall
that the task of $R_\gamma$ is to generate arbitrary amounts of
alphabet symbols $t_i^+$ and $t_i^-$ for all $t_i$ occurring in
$\gamma$. We define the set $T_\gamma$ collecting all $t_i^+$ and
$t_i^-$ corresponding to the inequalities appearing in $\gamma$:
\begin{align*}
  T_\gamma & \defeq 
  \begin{cases}
    \{ t_i^+, t_i^- \} & \text{if } \gamma = t_i\\
    T_\alpha \cup T_\beta & \text{if } \gamma = \alpha \wedge \beta \text{ or }
    \gamma = \alpha \vee \beta.
  \end{cases}
\end{align*}
We can now redefine $R_\gamma$ to be the regular grammar corresponding
to the following NFA, for which clearly $\lan(R_\gamma) =
T_\gamma^\odot$ holds:
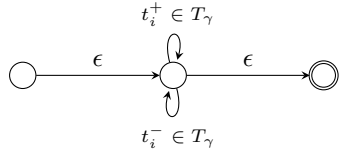
\begin{center}
  \begin{tikzpicture}[node distance=2cm]    
    
    
    \node (P1) [state] {$ $};
    \node (Q1) [right of=P1, state] {$ $};
    \node (R1) [right of=Q1, state, accepting] {$ $}; 
    \path[->]
    (P1) edge node[above] {$\epsilon$} (Q1)
    (Q1) edge node[above] {$\epsilon$} (R1)
    (Q1) edge [loop above] node[above right, xshift=-0.55cm] {\scriptsize{$t_i^+\in T_\gamma$}} (Q1)
    (Q1) edge [loop below] node[below right, xshift=-0.55cm] {\scriptsize{$t_i^-\in T_\gamma$}} (Q1);
  \end{tikzpicture}
\end{center}

Next, we describe an inductive procedure that when completed yields an
NFA that corresponds to a regular grammar whose language is equivalent
to $\lan(H,F_\psi)$. The procedure constructs in every iteration an
NFA with a unique incoming and outgoing state and labels every
transition with the gadget that should replace this transition in the
next iteration, or with an alphabet letter if no more replacement is
required. The initial such NFA is the following:
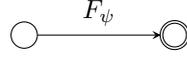
\begin{center}
  \begin{tikzpicture}[node distance=2cm]      
    \node (P0) [state] at (0,0) {$ $}; \node (Q0) [right of=P0, state,
      accepting]{$ $};
    
    \path[->] (P0) edge node[above] {$F_{\psi}$} (Q0);
  \end{tikzpicture}
\end{center}
In the induction step, the rewriting of a transition labelled with
$F_\gamma$ depends on the logical connective. A conjunction $\gamma =
\alpha \wedge \beta$ is replaced by sequential composition:
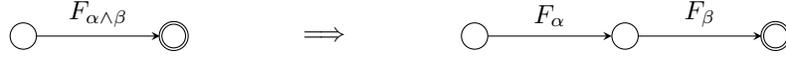
\begin{center}
  \begin{tikzpicture}[node distance=2cm]    
    \node (P0) [state] at (0,0) {$ $};
    \node (Q0) [right of=P0, state, accepting]{$ $};
    \path[->]
    (P0) edge node[above] {$F_{\alpha \land \beta}$} (Q0);
    
    \node (IM) [right of=Q0] {$\implies$};
    
    \node (P1) [right of=IM, state] {$ $};
    \node (Q1) [right of=P1, state] {$ $};
    \node (R1) [right of=Q1, state, accepting] {$ $}; 
    \path[->]
    (P1) edge node[above] {$F_{\alpha}$} (Q1)
    (Q1) edge node[above] {$F_{\beta}$} (R1);
  \end{tikzpicture}
\end{center}
Thus, the outgoing state of the gadget $F_\alpha$ connects to the
incoming state of the gadget $F_\beta$. In the case of a disjunction
$\gamma = \alpha \vee \beta$, the transition is rewritten into three paths that, informally speaking, correspond to possible truth
assignments to the subformulas $F_\alpha$ and $F_\beta$. If the
inequalities appearing in $\alpha$ are allowed to receive arbitrary
values, the transition labelled with $F_{\alpha \vee \beta}$ is
replaced by the sequential composition of two gadgets, $R_\alpha$ and
$F_\beta$; the other cases are treated in the same way:
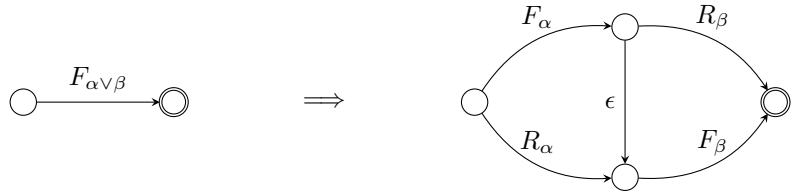
\begin{center}
  \begin{tikzpicture}[node distance=2cm]    
    \node (P0) [state] at (0,0) {$ $};
    \node (Q0) [right of=P0, state, accepting]{$ $};
    \path[->]
    (P0) edge node[above] {$F_{\alpha \lor \beta}$} (Q0);
    
    \node (IM) [right of=Q0] {$\implies$};
    
    \node (P1) [right of=IM, state] {$ $};
    \node (Q2) [right of=P1] {};
    \node (R1) [right of=Q2, state, accepting] {$ $};
    \node (Q1) [above of=Q2, yshift=-1cm, state] {$ $};
    \node (Q3) [below of=Q2, yshift=1cm, state] {$ $};
    
    \path[->]
    (P1) edge [bend left] node[above, yshift=0.03cm] {$F_{\alpha}$} (Q1)
    (Q1) edge [bend left] node[above] {$R_{\beta}$} (R1)
    (Q1) edge node[left] {$\epsilon$} (Q3)
    (P1) edge [bend right]node[above, yshift=0.03cm] {$R_{\alpha}$} (Q3)
    (Q3) edge [bend right]node[above] {$F_{\beta}$} (R1);
  \end{tikzpicture}
\end{center}
Once some $F_{t_i}$ is reached, it gets replaced by the empty word:
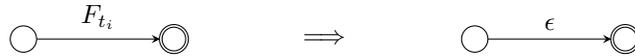
\begin{center}
  \begin{tikzpicture}[node distance=2cm]    
    \node (P0) [state] at (0,0) {$ $};
    \node (Q0) [right of=P0, state, accepting]{$ $};
    \path[->]
    (P0) edge node[above] {$F_{t_i}$} (Q0);
    
    \node (IM) [right of=Q0] {$\implies$};
    
    \node (P1) [right of=IM, state] {$ $};
    \node (Q1) [right of=P1, state, accepting] {$ $};
   \path[->]
   (P1) edge node[above] {$\epsilon$} (Q1) ;  
  \end{tikzpicture}
\end{center}
Moreover, if $R_\gamma$ is reached it gets replaced by the gadget
described earlier. From this construction, it is clear that the
regular grammar that can be obtained from the resulting NFA exactly
generates $\lan(H,F_\psi)$, and that in the following we may assume
that $H$ is regular.

\begin{figure}[t]
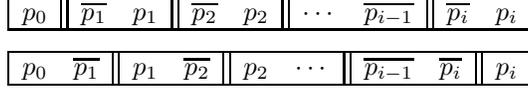

  \begin{center}
    \begin{tabular}{|c||cc||cc||cc||cc|}
      \hline
      $p_0$ & $\overline{p_1}$ & $p_1$ & $\overline{p_2}$ & $p_2$ & $\cdots$
      & $\overline{p_{i-1}}$ & $\overline{p_i}$ & $p_i$\\
      \hline
    \end{tabular}\\\vspace*{0.25cm}
    \begin{tabular}{|cc||cc||cc||cc||c|}
      \hline
      $p_0$ & $\overline{p_1}$ & $p_1$ & $\overline{p_2}$ & $p_2$ & $\cdots$
      & $\overline{p_{i-1}}$ & $\overline{p_i}$ & $p_i$\\
      \hline
    \end{tabular}
  \end{center}
  \caption{Illustration of the pairing of alphabet symbols in $C_\ell$
    (above) and $C_r$ (below).}\label{fig:pairing}
\end{figure}

We now turn towards showing how $G$ can be made regular. Even though
the structure of $G$ already appears to be regular, note that we use
the construction of Remark~\ref{rmk:binary-encoding} in order to
encode the constant $c$ in binary. This is not possible in the case of
regular commutative grammar, at least not in an obvious way. However,
we can use an interplay between $G$ and $H$ in order to, informally
speaking, force $G$ to produce alphabet symbols in exponential
quantities. To this end, we introduce additional alphabet symbols and
define $\Gamma_i\defeq \{
p_0,\overline{p_1},p_1,\ldots,\overline{p_i}, p_i \}$ for every $i\in
\N$. Before formally providing the construction in
Lemma~\ref{lem:exponentiating} below, let us discuss how we can
achieve our goal on an informal level. Suppose we wish to produce a
word $w \in \Gamma_i^\odot$ such that $w(p_i)= 2^i \cdot p_0$. One way
to obtain a language that contains such a word is to pair alphabet
symbols $\overline{p_j}$ and $p_{j}$ and to produce two symbols $p_j$
every time some $\overline{p_j}$ is non-deterministically
produced. The pairing is illustrated in the top of
Figure~\ref{fig:pairing}, and, more formally, such a language can be
generated by the following regular grammar: $C_\ell\defeq
(\{S_\ell\},\Gamma_i,S_\ell,P_\ell)$, where
\begin{align*}
  S_\ell & \rightarrow \epsilon\\
  S_\ell & \rightarrow S_\ell p_0\\
  S_\ell & \rightarrow S_\ell \overline{p_j}p_jp_j  & (0\le j\le i).
\end{align*}
Clearly, we can find some $w\in \lan(C_\ell)$ such that
\[
w(p_i)=2\cdot w(\overline{p_i}) = 2 \cdot w(p_{i-1}) = 2^2 \cdot
w(\overline{p_{i-1}}) = \cdots = 2^i \cdot w(p_0).
\]
Such a $w$ implicitly requires another pairing, namely that
$w(\overline{p_{j+1}})=w(p_{j})$ for all $0\le j<i$, which is
illustrated in the bottom of Figure~\ref{fig:pairing}. If we can rule
out all words of $\lan(C_\ell)$ that violate this pairing, we obtain a
language containing the desired $w\in \Gamma_i^\odot$ such that
$w(p_i)=2^i\cdot w(p_0)$. This is the task of the regular grammar
$C_r$ constructed in the following lemma.
\begin{lemma}\label{lem:exponentiating}
  For every $i\in \N$, there are logarithmic-space computable regular
  commutative grammars $C_\ell$ and $C_r$ such that:
  \begin{enumerate}[(i)]
  \item $\lan(C_\ell) = \{ w\in \Gamma_i^\odot : w(p_j) = 2\cdot w(\overline{p_j})
    $ for every $1\le j\le i \}$; and
  \item $\lan(C_r) = \{ w\in \Gamma_i^\odot : w(p_{j}) \neq
    w(\overline{p_{j+1}})$ for some $0\le j<i \}$.
  \end{enumerate}
  In particular, for every $v\in \lan(C_\ell)\setminus \lan(C_r)$,
  $v(p_i) = 2^i\cdot v(p_0)$.
\end{lemma}
\begin{proof}
  Regarding Part~(i), clearly $C_\ell$ as defined above has the
  desired properties. In order to prove Part~(ii), we define $C_r
  \defeq (N_r, \Gamma_i, S_r, P_r)$, where $N_r \defeq \{ S_r\} \cup
  \{ N_j, \overline{N_{j+1}} : 0\le j < i \}$ and
  \begin{align*}
    S_r & \rightarrow S_r p_i\\
    S_r & \rightarrow N_j p_j & S_r & \rightarrow \overline{N_{j+1}}~\overline{p_{j+1}}
    & (0\le j < i)\\
    N_j & \rightarrow N_jp_j & \overline{N_{j+1}} & \rightarrow \overline{N_{j+1}}
    \overline{p_{j+1}} & (0\le j< i)\\
    N_j & \rightarrow N_jp_j\overline{p_{j+1}} & \overline{N_{j+1}} & \rightarrow 
    \overline{N_{j+1}} p_{j}\overline{p_{j+1}} & (0\le j<i)\\
    N_j & \rightarrow N_jp_g & \overline{N_j} & \rightarrow \overline{N_j}p_g & 
    (0\le g,j< i, g \neq j)\\
    N_j & \rightarrow N_j\overline{p_{g+1}} & \overline{N_j} & \rightarrow 
    \overline{N_j} \overline{p_{g+1}} & (0\le g,j< i, g \neq j)\\
    N_j & \rightarrow \epsilon & \overline{N_j} & \rightarrow \epsilon &
    (0\le j<i).
  \end{align*}
  Informally speaking, after non-deterministically producing alphabet
  symbols $p_i$ starting from $S_r$, we can then non-deterministically
  choose an index $0\le j<i$ such that either $p_j >
  \overline{p_{j+1}}$ (when switching to $N_j$) or $\overline{p_{j+1}}
  > p_j$ (when switching to $\overline{N_j}$), and for any choice of
  $j$ all other alphabet symbols $p_g$ and $\overline{p_{g+1}}$ such
  that $g\neq j$ can be produced in arbitrary quantities. It is easily
  checked that $\lan(C_r)$ has the desired properties. Now, we have
  \begin{align*}
    & v\in \lan(C_\ell)\setminus \lan(C_r)\\
    \iff & v\in \lan(C_\ell) \cap  \overline{\lan(C_r)}\\
    \iff & v\in \lan(C_\ell) \cap \{ w \in \Gamma_i^\odot : w(p_j) = 
    w(\overline{p_{j+1}}) \text{ for all } 0 \le j< i \}\\
    \iff & v\in \{ w\in \Gamma_i^\odot : w(p_{j+1}) = 2\cdot w(\overline{p_{j+1}}) 
    \text { and }
    w(p_j) =  w(\overline{p_{j+1}}) \text{ for all } 0 \le j< i\}\\
    \implies &  v\in \{ w \in \Gamma_i^\odot : w(p_{j+1}) = 2\cdot w(p_j)
   \text{ for all } 0\le j<i \}\\
   \implies & v(p_i) = 2^i\cdot v(p_0).
  \end{align*}\end{proof}
%

Let $\Sigma=\{ t_1^+,t_1^-,\ldots, t_k^+, t_k^-\}$ be as defined in
the previous section. The following corollary is an immediate
consequence of Lemma~\ref{lem:exponentiating} and enables us to
construct an exponential number of $t_i^+$ and $t_i^-$.
\begin{corollary}\label{cor:extended-exponentiating}
  For every $i\in \N$, there are logarithmic-space computable regular
  commutative grammars $C_\ell^\Sigma$ and $C_r^\Sigma$ over $\Sigma \cup
  \Gamma_i$ such that
  \begin{enumerate}[(i)]
  \item $\lan(C_\ell^\Sigma) = \lan(C_\ell) \cdot \Sigma^\odot \cap \{ w \in
    {(\Sigma\cup \Gamma_i)}^\odot : w(t_j^+) = w(t_j^-) = w(p_i) \}$;
    and
  \item $\lan(C_r^\Sigma) = \lan(C_r) \cdot \Sigma^\odot$.
  \end{enumerate}
  where $C_\ell$ and $C_r$ are defined as in Lemma~\ref{lem:exponentiating}.
\end{corollary}

Recall that $H$ already is a regular commutative grammar, and let $c$
be the constant from~(\ref{eqn:constant-c}) and $j\defeq \log c$. We
can now define regular versions $G^r$ and $H^r$ over $\Sigma \cup
\Gamma_{j}$ of $G$ and $H$, respectively, such that $\lan(G^r)
\subseteq \lan(H^r)$ iff $\lan(G)\subseteq \lan(H)$. Let $u$ and $v_i$
be defined as in~(\ref{eqn:u-def}) and~(\ref{eqn:v-i-def}), and let
$C_\ell^\Sigma$ and $C_r^\Sigma$ be as defined in
Corollary~\ref{cor:extended-exponentiating} for the alphabet
$\Gamma_j$, the axiom of $G^r$ is $S_G^r$ and the transitions of $G^r$
are given by
\begin{align*}
  S_G^r   & \rightarrow X p_0 u  &   
  X & \rightarrow X p_0 v_j  & (1\le j\le m)\\
  X   & \rightarrow C_\ell^\Sigma
\end{align*}
Moreover, $H^r$ is the regular commutative grammar such that
\[
\lan(H^r) = \lan(C_r^\Sigma) \cup \lan(H)\cdot \Gamma_{j}^\odot.
\]
The correctness of the construction can be seen as follows. Let $w\in
\lan(G^r)$, we distinguish two cases:
\begin{enumerate}[(i)]
\item If $w(p_i)\neq c\cdot w(p_0)$ then $w\in \lan(C_r^\Sigma)
  \subseteq \lan(H^r)$.
\item Otherwise, $w(p_i)= c \cdot w(p_0)$ 
  and $w\not\in \lan(C_r^\Sigma)$. Consequently, $w \in
  \lan(H^r)$ iff $w\in\lan(H)\cdot \Gamma_{j}^\odot$ thus $\pi_\Sigma(w)\in \lan(H)$. 
  But we know that $\pi_\Sigma(w)\in \lan(G)$.
\end{enumerate}
Concluding, we have $\lan(G^r)\subseteq \lan(H^r)$ if
$\lan(G)\subseteq \lan(H)$. The implication in the opposite direction
is obvious, which completes our proof.


\section{Improved Complexity Bounds for Language Equivalence for Exponent-Sensitive Commutative Grammars}\label{sec:upper}
In this section, we turn towards the equivalence problem for
exponent-sensitive commutative grammars and prove
Theorem~\ref{thm:main-expsens}. Hardness for \coNEXP\ of this problem
directly follows from Theorem~\ref{thm:main}, since regular
commutative grammars are a subclass of exponent-sensitive
grammars. Hence, here we show that the problem can be decided in
\co-\ComplexityFont{2}\NEXP, thereby improving the
\ComplexityFont{2}\EXPSPACE\ upper bound from~\cite{MW13}. As stated
in Section~\ref{sec:preliminaries}, commutative words on the left-hand
sides of the productions of exponent-sensitive commutative grammar as
defined in~\cite{MW13} are encoded in binary. 

It is sufficient to show that inclusion between exponent-sensitive
commutative grammars can be decided in \co-\ComplexityFont{2}\NEXP.
To this end, we follow an approach proposed by Huynh used to show that
inclusion of context-free commutative grammars is
in~\coNEXP~\cite{Huy85}. Let $G$ and $H$ be exponent-sensitive
commutative grammars. The starting point of Huynh's approach is to
derive bounds on the size of a commutative word witnessing
non-inclusion via the semi-linear representation of the reachability
sets of $G$ and $H$. For exponent-sensitive commutative grammars,
$\reach(G)$ and $\reach(H)$ are shown semi-linear with a
representation size doubly exponential in~$\#G + \#H$
in~\cite{MWtr13}, and this representation is also computable in
doubly-exponential time. Given semi-linear sets $M$ and $N$ such that
$M\setminus N$ is non-empty, Huynh shows in~\cite{Huy86} that there is
some $\vec{v}\in M\setminus N$ whose bit-size is polynomial in $\#M +
\#N$. Consequently, if $\lan(G) \not \subseteq \lan(H)$ then the
binary representation of some word $w\in \lan(G)\setminus \lan(H)$ has
size bounded by $2^{2^{p(\#G + \#H)}}$ for some polynomial $p$. Since
the word problem for exponent-sensitive commutative grammars is
in~\PSPACE, deciding $\lan(G)\subseteq \lan(H)$ is in
\ComplexityFont{2}-\EXPSPACE, as observed
in~\cite[Thm.~5.5]{MWtr13}. Now comes the second part of Huynh's
approach into play. In~\cite{Huy85}, a Carath{\'e}odory-type theorem
for semi-linear sets is established: given a linear set
$M=L(\vec{b},P)\subseteq \N^m$, Huynh shows that $M=\bigcup_{i\in
  I}L(\vec{b}_i,P_i) $, where $\vec{b}_i\in L(\vec{b},P)$, each
$\vec{b}_i$ has bit-size polynomial in $\#M$, and $P_i\subseteq P$ has
full column rank and hence in particular $\abs{P_i}\le m$. The key
point is that deciding membership in a linear set with such properties
obviously is in~\P\ using Gaussian elimination, and that we can show
that a semi-linear representation of $\reach(G)$ and $\reach(H)$ in
which every linear set has those properties is computable in
deterministic doubly-exponential time in $\#G + \#H$. Consequently, a
\co-\ComplexityFont{2}\NEXP\ algorithm to decide $\lan(G) \subseteq
\lan(H)$ can initially guess a word $w$ whose representation is
doubly-exponential in $\#G + \#H$, then compute the semi-linear
representations of $\reach(G)$ and $\reach(H)$ in the special form of
Huynh, and check in polynomial time in $\#w$ that $w$ belongs to
$\lan(G)$ and not to $\lan(H)$. We now proceed with the formal
details.

Subsequently, let $s\defeq \#G$, $t\defeq \#H$ and
$\lan(G),\lan(H)\subseteq \Sigma^\odot$. We begin with stating the
relevant facts about the semi-linear representation of the
reachability set of exponent-sensitive commutative grammars. The
subsequent proposition is derived from~\cite[Lem.~5.4]{MWtr13}, which
is stated in terms of generalised communication-free Petri nets, but
as argued in the proof of~\cite[Thm.~6.1]{MWtr13}, there is a
logarithmic-space reduction from exponent-sensitive commutative
grammars to such Petri nets which preserves reachability sets, and
hence allows us to apply~\cite[Lem.~5.4]{MWtr13}.
\begin{proposition}[\cite{MWtr13}]\label{prop:exp-sens-semi-lin}
  There exists a fixed polynomial $p$ such that the
  reachability set $\reach(G) = \bigcup_{i \in I} L(\vec{b}_i,Q_i)$ is
  computable in \DTIME($2^{2^{\poly(s)}}$) such that for every $i\in
    I$,
  \begin{itemize}
  \item $\abs{I}\le 2^{2^{p(s)}}$ and $\abs{Q_i}\le 2^{p(s)}$; and
  \item $\#\vec{b}_i \le p(s)$ and $\#\vec{q}\le p(s)$ for every
    $\vec{q}\in Q_i$.
  \end{itemize}
\end{proposition}
Next, we introduce Huynh's decomposition of linear sets as described
above. The following proposition is a consequence and a summary of
Proposition~2.6 and Lemmas~2.7 and~2.8 in~\cite{Huy85}.
\begin{proposition}[\cite{Huy85}]\label{prop:caratheodory}
  Let $M=L(\vec{b},Q)$ be a linear set. There is a fixed polynomial
  $p$ such that $M=\bigcup_{i\in I} M_i$ and for every $i\in
  I$, $M_i = L(\vec{b}_i, Q_i)$ with
  \begin{itemize}
  \item $\vec{b}_i \in L(\vec{b},Q)$ and $\# \vec{b}_i \le p(\# M)$;
    and
  \item $Q_i\subseteq Q$ is has full column rank and $\abs{Q_i}=
    \rank(Q)$.
  \end{itemize}
\end{proposition}
Subsequently, for a given $M=L(\vec{b},Q)$, whenever $\bigcup_{i\in I}
L(\vec{b}_i,Q_i)$ has the properties described in
Proposition~\ref{prop:caratheodory}, we say that it is the
Huynh representation of $M$.

\begin{lemma}\label{lem:caratheodory-complexity}
  Let $M=L(\vec{b},Q)$ be a linear set. The Huynh representation of
  $M$ can be computed $\DTIME(2^{\poly(\#M)})$.
\end{lemma}
\begin{proof}
  Let $p$ be the polynomial from Proposition~\ref{prop:caratheodory}.
  First, we compute the set of $\vec{b}_i$ as follows: we enumerate
  all candidates $\vec{b}_i$ such that $\#\vec{b}_i\le p(\#M)$, there
  is at most an exponential number of them. For every candidate we
  check if $\vec{b}_i\in L(\vec{b},Q)$, which can be done in
  $\NP$. Next, we enumerate all subsets $Q_i\subseteq Q$ of full
  column rank and cardinality $\rank(Q)$, again there are at most
  exponentially many of them. Finally, we output the all possible
  combinations of the $\vec{b}_i$ with the $Q_i$.
\end{proof}
\begin{lemma}
  The Huynh representation of $\reach(G)$ can be computed in
  $\DTIME(2^{2^{\poly(s)}})$.
\end{lemma}
\begin{proof}
  First, we apply Proposition~\ref{prop:exp-sens-semi-lin} in order to
  compute a semi-linear representation $\bigcup_{i\in
    I}L(\vec{b}_i,Q_i)$ of $\reach(G)$ such that $\abs{I}\le
  2^{2^{p(s)}}$, $\abs{Q_i}\le 2^{p(s)}$, and $\#\vec{b}_i \le p(s)$
  and $\#\vec{q}\le p(s)$ for every $\vec{q}\in Q_i$ for some fixed
  polynomial $p$. By
  Lemma~\ref{lem:caratheodory-complexity}, from every
  $M_i=L(\vec{b}_i,Q_i)$ we can compute an equivalent
  Huynh representation $N_i=\bigcup_{j\in J_i}
  L(\vec{c}_{i,j},R_{i,j})$ of $M_i$ in
  $\DTIME(2^{\poly(\#M_i)})=\DTIME(2^{2^{\poly(s)}})$. Thus, the
  overall procedure also runs in $\DTIME(2^{2^{\poly(s)}})$.
\end{proof}

As the final ingredient, we state Huynh's result that whenever
inclusion between two semi-linear sets does not hold then there exists
a witness of polynomial bit-size.
\begin{proposition}[\cite{Huy86}]\label{prop:semi-lin-inclusion}
  Let $M,N \subseteq \N^m$ be semi-linear sets. There is a fixed
  polynomial $p$ such that whenever $M\not \subseteq N$ then
  there exists some $\vec{v}\in M\setminus N$ such that $\size
  \vec{v}\le p(\size M + \size N)$.
\end{proposition}

We are now fully prepared to prove the main statement of this section,
which immediately yields the upper bound for
Theorem~\ref{thm:main-expsens}.
\begin{proposition}
  Deciding $\lan(G)\subseteq \lan(H)$ is in
  \co-\ComplexityFont{2}$\NEXP$.
\end{proposition}
\begin{proof}
  We describe a \co-\ComplexityFont{2}\NEXP-algorithm. First, by
  combining Proposition~\ref{prop:exp-sens-semi-lin} with
  Proposition~\ref{prop:semi-lin-inclusion}, if $\lan(G)\not \subseteq
  \lan(H)$ then there is some $w\in \Sigma^\odot$ such that $\#w \le
  2^{2^{p(s+t)}}$ for some fixed polynomial $p$. The algorithm
  non-deterministically chooses such a $w$. Now the algorithm computes
  the Huynh representations of $\reach(G)$ and $\reach(H)$ in
  $\DTIME(2^{2^{\poly(s+t)}})=\DTIME(\poly(\#w))$. For every linear
  set $M=L(\vec{b}, Q)$ in the Huynh representation of $\reach(G)$ and
  $\reach(H)$, $\# \vec{b} \le p(s+t)$, $\#\vec{q}\le p(s+t)$ for all
  $\vec{q}\in Q$ and some fixed polynomial $p$, and $Q$ has full
  column rank and hence $\abs{Q}\le \abs{\Sigma}$. Thanks to those
  properties, $w\in L(\vec{b},Q)$ can be decided in
  $\DTIME(\poly(\#M))$ using Gaussian elimination. Consequently,
  checking $w \in \lan(G)\setminus \lan(H)$ can be performed in
  $\DTIME(\poly(\#w))$.
\end{proof}


\section{Applications to Equivalence Problems for Classes of Petri Nets, BPPs and
  Reversal-Bounded Counter Automata}
Here, we discuss immediate corollaries of Theorems~\ref{thm:main}
and~\ref{thm:main-expsens} for equivalence problems for various
classes of Petri nets, basic parallel process nets (BPP-nets) and
reversal-bounded counter automata.

It has, for instance, been observed in~\cite{Esp97,Yen97,MW15} that
context-free commutative grammars can be seen as notational variants
of communication-free Petri nets and BPP-nets. This allows for
transferring results on standard decision problems between these
formalisms. We do not formally introduce communication-free Petri nets
and BPP-nets here. Informally speaking, in those nets non-terminal and
terminal symbols of commutative context-free grammars correspond to
places in those nets, where every transition can remove a token from
at most one place, and where tokens cannot be removed from places
corresponding to terminal symbols. The equivalence problem for
communication-free Petri nets and BPP-nets is to decide
whether the set of reachable markings of two given nets coincides. In
particular, this requires that both nets have the same set of
places. Via a reduction to the equivalence problem for context-free
commutative grammars, it is possible to obtain a \coNEXP-upper bound
for the equivalence problem of communication-free Petri nets and
BPP-nets~\cite{MW15}. On the other hand, Theorem~\ref{thm:main}
together with the strengthened construction given
in~(\ref{eqn:bpp-reduction}) yields a matching lower bound.

\begin{theorem}
  The equivalence problem for communication-free Petri nets and
  BPP-nets is \coNEXP-complete.
\end{theorem}

As already briefly mentioned in Section~\ref{sec:upper},
exponent-sensitive commutative grammars are closely related to
so-called generalised communication-free Petri nets, and in fact
inter-reducible with them~\cite[Thm.~6.1]{MWtr13}. Similarly to
communication-free Petri nets, transitions of generalised
communication-free Petri nets may only remove tokens from one place,
however they are allowed to remove an arbitrary number of them, not
just one. For the sake of completeness, let us state
Theorem~\ref{thm:main-expsens} in terms of generalised
communication-free Petri nets.
\begin{theorem}
  The equivalence problem for generalised communication-free Petri
  nets is \coNEXP-hard and in \co-\ComplexityFont{2}\NEXP.
\end{theorem}

We now turn towards the equivalence problem for reversal-bounded
counter automata. For our purposes, it is sufficient to introduce
reversal-bounded counter automata on an informal level, formal
definitions can be found in~\cite{Ibar78}, see also~\cite{Ibar14} for
a recent survey on decision problems for reversal-bounded counter
automata. A counter automaton comprises a finite-state controller with
a finite number of counters ranging over the natural numbers that can
be incremented, decremented or tested for zero along a transition. A
classical result due to Minsky states that reachability in counter
automata is undecidable already in the presence of two
counters~\cite{Min61}. One way of overcoming this problem is to bound
the number of times a counter is allowed to switch between increasing
and decreasing mode. Consider, for instance, a counter of some counter
automaton whose values along a run are $0,1,2,3,5,4, 4, 3, 3, 4, 5,
6$. On this run, the counter reverses its mode twice, once from
incrementing to decrementing mode (when decrementing from 5 to 4), and
one more time from decrementing to incrementing mode (when,
thereafter, incrementing from 3 to 4). In this example, the number of
reversals of the counter is bounded by two. A $k$-reversal bounded
counter automaton is a counter automaton whose counters are only
allowed to have at most $k$ reversals along a run.

Ibarra~\cite{Ibar78} has shown that the sets of reachable
configurations of reversal-bounded counter automata are effectively
semi-linear. Hague and Lin~\cite{HL11} showed that from a
reversal-bounded counter automaton one can construct in polynomial
time an open existential Presburger formula $\varphi(\vec{x})$
defining the set of reachable configurations, i.e., the sets of
counter values with which a target control state is reached starting
in an initial configuration. The equivalence problem for reversal
bounded counter automata over the same number of counters is to decide
whether their reachability sets are the same.

An application of the result of Hague and Lin combined with
Proposition~\ref{prop:pa-pi-2-complexity} immediately yields a
\coNEXP-upper bound for the equivalence problem. Given two
reversal-bounded counter automata whose reachability sets are defined
by existential Presburger formulas $\varphi(\vec{x})$ and
$\psi(\vec{x})$, respectively, their reachability set is equivalent
iff $\phi\defeq \forall \vec{x}. \phi(\vec{x}) \leftrightarrow
\psi(\vec{x})$ is valid. Since $\phi$ is a $\Pi_2$-sentence of
Presburger arithmetic, Proposition~\ref{prop:pa-pi-2-complexity}
yields a \coNEXP-upper bound for the equivalence problem. On the other
hand, Theorem~\ref{thm:main} gives that equivalence is already
\coNEXP-hard for regular commutative grammars. Now the latter can
immediately be simulated by a $0$-reversal bounded counter automaton
by introducing one counter for each alphabet symbol in $\Sigma$,
treating non-terminal symbols as control states, and incrementing the
counter corresponding to some $a\in \Sigma$ whenever a production
$V\rightarrow Wa$ is simulated. Consequently, we have proved the
following theorem.

\begin{theorem}
  The equivalence problem for reversal-bounded counter automata is
  \coNEXP-complete, and in particular \coNEXP-hard for $0$-reversal
  bounded counter automata with no zero tests whose constants are
  encoded in unary.
\end{theorem}


\section{Conclusion}
In this paper, we showed that language inclusion and equivalence for
regular and context-free commutative grammars are \coNEXP-complete,
resolving a long-standing open question posed by Huynh~\cite{Huy85}.
Our lower bound also carries over to the equivalence problem for
exponent-sensitive commutative grammars, for which we could also
improve the \ComplexityFont{2}-\EXPSPACE-upper bound~\cite{MW13} to
\co-\ComplexityFont{2}$\NEXP$. The precise complexity of this problem
remains an open problem of this paper. An overview over the complexity
of word and equivalence problems for commutative grammars together
with references to the literature is provided in
Table~\ref{tab:complexity}.
\begin{table}[t]
  \begin{center}
    \renewcommand{\arraystretch}{1}
    \begin{tabular}{|r|c|c|}
      \hline
      com.\ grammar
       & word problem & language equivalence\\
      \hline\hline
      type-0 & \EXPSPACE-h.~\cite{Lipt76}, $\in \mathbf{F}_{\omega^3}$~\cite{LS15}
      & undecidable~\cite{Hack76}\\
      \hline
      cont.-sensitive & \PSPACE-complete~\cite{Huy83} & undecidable~\cite{Huy85}\\
      \hline
      exp.-sensitive & \PSPACE-complete~\cite{MW13} & 
      \mbox{\coNEXP-h.,
      $\in \co$-$\ComplexityFont{2}\NEXP$}\\
      \hline
      context-free & &  \\
      \cline{1-1}
      regular & \multirow{-2}{*}{\NP-complete~\cite{Huy83,Esp97}} & 
      \multirow{-2}{*}{\coNEXP-complete} \\
      \hline
    \end{tabular}
  \end{center}  
  \caption{Complexity of the word and the equivalence problem for
    classes of commutative grammars.}\label{tab:complexity}
\end{table}

It is interesting to note the non-monotonic behaviour of regular
grammars with respect to the complexity of the equivalence problem.
In the non-commutative setting, language equivalence
is \PSPACE-complete, and hardness even holds when the number of
alphabet symbols is fixed. In contrast, in the commutative setting
language equivalence is $\ComplexityFont{\Pi_2^\P}$-complete when the
number of alphabet symbols is fixed~\cite{KT10,Kop15}, and
$\coNEXP$-complete for an alphabet of arbitrary size as shown in this
paper.

One major open problem related to the problems discussed in this paper
is weak bisimilartiy between basic parallel processes. This problem is
not known to be decidable
and \PSPACE-hard~\cite{Srba03}. Unfortunately, it does not seem
possible to adjust the construction of our \coNEXP-lower bound to also
work for weak bisimulation.

\subsection*{Acknowledgements.} We would like to thank Matthew Hague for
clarifying some questions regarding reversal-bounded counter automata.


\bibliographystyle{amsalpha}
\bibliography{bibliography}

\end{document}